\documentclass[conference]{IEEEtran}
\IEEEoverridecommandlockouts
\usepackage{cite}
\usepackage{amsmath,amssymb,amsfonts}
\usepackage{algorithmic}
\usepackage{graphicx}
\usepackage{textcomp}
\usepackage{xcolor}
\def\BibTeX{{\rm B\kern-.05em{\sc i\kern-.025em b}\kern-.08em
    T\kern-.1667em\lower.7ex\hbox{E}\kern-.125emX}}

\usepackage{amssymb}
\usepackage{amsmath}
\usepackage{mathrsfs}
\usepackage{amsthm}
\usepackage{verbatim}
\usepackage{dsfont}

\usepackage{caption}

\newtheorem{theorem}{Theorem}
\newtheorem{lemma}{Lemma}

\theoremstyle{definition}
\newtheorem{remark}{Remark}
\newtheorem{definition}{Definition}
\newtheorem{example}{Example}

\usepackage{algorithm}
\usepackage{bbm}
\usepackage{hyperref}
\hypersetup{colorlinks = true
            linkcolor=red,
            anchorcolor=blue,
            citecolor=green}
              
\newcommand{\MF}{\mathcal{F}}
\newcommand{\MI}{\mathcal{I}}
\newcommand{\mM}{\mathcal{M}}
\newcommand{\FF}{\mathbb{F}}
\newcommand{\RR}{\mathbb{R}}
\newcommand{\NN}{\mathbb{N}}
\newcommand{\Ba}{\mathds{1}}

\begin{document}

\title{The Complete SC-Invariant Affine Automorphisms of Polar Codes}

 \author{
   \IEEEauthorblockN{Zicheng Ye\IEEEauthorrefmark{2}\IEEEauthorrefmark{3}, 
   Yuan Li\IEEEauthorrefmark{2}\IEEEauthorrefmark{3}\IEEEauthorrefmark{1}, 
   Huazi Zhang\IEEEauthorrefmark{1}, 
   Rong Li\IEEEauthorrefmark{1}, 
   Jun Wang\IEEEauthorrefmark{1}, 
   Guiying Yan\IEEEauthorrefmark{2}\IEEEauthorrefmark{3}, 
   and Zhiming Ma\IEEEauthorrefmark{2}\IEEEauthorrefmark{3} }
  \IEEEauthorblockA{\IEEEauthorrefmark{2}
                     University of Chinese Academy of Sciences}
   \IEEEauthorblockA{\IEEEauthorrefmark{3}
                     Academy of Mathematics and Systems Science, CAS }
  \IEEEauthorblockA{\IEEEauthorrefmark{1}
                     Huawei Technologies Co. Ltd.}
    Email: \{yezicheng, liyuan2018\}@amss.ac.cn, \{zhanghuazi, lirongone.li, justin.wangjun\}@huawei.com,\\
           yangy@amss.ac.cn, mazm@amt.ac.cn  }

\maketitle

\begin{abstract}
Automorphism ensemble (AE) decoding for polar codes was proposed by decoding permuted codewords with successive cancellation (SC) decoders in parallel and hence has lower latency compared to that of successive cancellation list (SCL) decoding. However, some automorphisms are SC-invariant, thus are redundant in AE decoding. In this paper, we find a necessary and sufficient condition related to the block lower-triangular structure of transformation matrices to identify SC-invariant automorphisms. Furthermore, we provide an algorithm to determine the complete SC-invariant affine automorphisms under a specific polar code construction.
\end{abstract}

\section{Introduction}
Polar codes \cite{b1} are proved to asymptotically achieve capacity on discrete binary memoryless symmetric (BMS) channels under SC decoding. To enhance the finite-length performance, SCL decoding was proposed in \cite{b2}. Moreover, cyclic redundancy check (CRC)-aided polar codes \cite{b3} achieve outstanding performance at short to moderate block lengths.

A substantial part of SCL decoding complexity and latency is related to path management, i.e., sorting and pruning paths according to path metric (PM). In order to reduce the latency, decoding under stage permutations on the factor graph was proposed in \cite{b4}. In \cite{b5}, AE decoding utilized more SC-variant automorphisms instead of only stage permutations to enhance error correcting performance. A key step in AE decoding is the identification and avoidance of SC-invariant automorphisms, which produce duplicate decoding results. The automorphisms formed by lower-triangular affine (LTA) transformations \cite{b6} were proved to be SC-invariant \cite{b5}. In \cite{b7} and \cite{b8}, the block lower-triangular affine (BLTA) group was proved to be the complete affine automorphism group of polar codes. BLTA transformations showed better performance under AE decoding \cite{b7} \cite{b9} \cite{b10}. In \cite{b10}, affine automorphism group was classified into equivalent classes, where each class will yield the same SC decoding result. Therefore selecting at most one automorphism from each equivalent class guarantees SC-variance. In contrast of AE decoding, some other applications require SC-invariant automorphisms. In \cite{b11}, $\frac{n}{4}$-cyclic shift permutations, which are SC-invariant, were proposed for implicit timing indication in Physical Broadcasting Channel (PBCH). In both applications, identifying SC-invariant automorphisms is a key step.

Some previous works attempt to identify SC-invariant affine automorphisms for general polar codes \cite{b5} \cite{b10}. However, given a specific code construction, SC-invariant automorphisms can not be completely found in \cite{b5}, \cite{b10}. In this paper, we identify and prove the complete SC-invariant affine automorphisms. For example, as shown in Table \ref{tab1} of section IV, the number of the complete SC-invariant affine automorphisms for (256,128) polar code is $21\times 2^{28}$ but only $3\times 2^{28}$ of them are founded in \cite{b10}.

The rest of this paper is organized as follows. In section II, we review polar codes and automorphism group. In section III, we provide a low complexity algorithm to distinguish SC-invariant affine automorphisms.  We further prove SC-invariant affine automorphism group is also of the form BLTA and provide an algorithm to determine it. In section IV, simulation results show distinguishing SC-invariant automorphisms can reduce redundancy in AE decoding. Finally, we draw some conclusions in section V.

\section{Preliminaries}

\subsection{Polar codes as monomial codes}
Let $F=\begin{bmatrix} 1&0 \\ 1&1 \end{bmatrix}$ and $G_m=F^{\otimes m}$, where $m$ is the code dimension. A polar code $(n=2^m,K)$ is generated by selecting $K$ rows of $G_m$. The set $\MI\subseteq \{0,1,...,n-1\}$ of indices of selected rows is the information set, and $\MF = \MI^c$ is the frozen set. Denote the polar code with information set $\MI$ by $C(\MI)$. 

Polar codes can be described as monomial codes \cite{b6}. The monomial set is
\[\mM=\{x_1^{g_1}...x_m^{g_m}|(g_1,...,g_m)^T\in\FF_2^m\},\]
and the evaluation of $g\in \mM$ is
\[\text{eval}(g) = (g(u))_{u\in\FF_2^m}.\]

Then each row of $G_m$ can be represented by $\text{eval}(g)$ for some $g\in\mM$. For example, assume $z\in \{0,1,...,2^m-1\}$, there is a unique binary representation $a=(a_1,...,a_m)^T$ of $2^m-z-1$, where $a_1$ is the least significant bit, such that
\[\sum_{i=1}^m 2^{i-1}(1-a_i) = z.\]
Then the evaluation of monomial $\text{eval}(x_1^{a_1}...x_m^{a_m})$ is exactly the $(2^m-z-1)$-th row of $G_m$. Therefore, the information set $\MI$ can be regarded as a subset of $\{0,...,n-1\}$ or a subset of $\mM$. As seen, the three representations, i.e., the number $z$, the binary representation of $2^m-1-z = (a_1,...,a_m)^T$ and the corresponding monomial $x_1^{a_1}...x_m^{a_m}$ all refer to the same thing.
 
Two monomials of the same degree are ordered as $x_{i_1}...x_{i_t}\preccurlyeq x_{j_1}...x_{j_t}$ if and only if $i_l \leq j_l$ for all $l\in\{1,...,t\}$, where we assume $i_1 <...< i_t$ and $j_1 <...< j_t$. This partial order is extended to monomials with different degrees through divisibility, namely $f \preccurlyeq g$ if and only if there is a divisor $g'$ of $g$ such that $f\preccurlyeq g'$.

An information set $\MI\subseteq \mathcal{M}$ is decreasing if $\forall g\preccurlyeq f$ and $f\in\MI$ we have $g \in \MI$. A decreasing monomial code $C(\MI)$ is a monomial code with a decreasing information set $\MI$. If the information set is selected according to the Bhatacharryya parameter, polar codes will be decreasing monomial codes \cite{b6}, \cite{b13}. In this way, polar codes can be generated by $\MI_{\text{min}}$, where the information set is the smallest decreasing set containing $\MI_{\text{min}}$. From now on, we always suppose $\MI$ is decreasing.

\subsection{Affine automorphism group}

Let $C$ be a decreasing monomial code with length $n$. A permutation $\pi$ in the symmetric group Sym$(n)$ is an automorphism of $C$ if for any codeword $c=(c_0,...,c_{n-1})\in C$, $\pi(c)=(c_{\pi(0)},...,c_{\pi(n-1)})\in C$. The automorphism group Aut$(C)$ is the subgroup of  Sym$(n)$ containing all automorphisms of $C$. 

Let $M$ be an $m\times m$ binary invertible matrix and $b$ be a length-$m$ binary column vector. The affine transformation $(M,b)$ permutes $a\in \FF_2^m$ to $Ma+b$. 

A matrix $M$ is lower-triangular if $M(i,i)=1$ and $M(i,j)=0$ for all $j>i$. The LTA group is the group of all affine transformations $(M,b)$ where $M$ is lower-triangular. Similarly, a matrix $M$ is upper-triangular if $M(i,i)=1$ and $M(i,j)=0$ for all $j<i$.

BLTA$([s_1,...,s_l])$ is a BLTA group of all affine transformations $(M,b)$ where $M$ can be written as a block matrix of the following form
\begin{equation}
\begin{bmatrix}
B_{1,1} & 0 & \cdots & 0 \\
B_{2,1} & B_{2,2} & \cdots & 0 \\
\vdots & \vdots  & \ddots & 0 \\
B_{l,1} & B_{l,2} & \cdots & B_{l,l}
\end{bmatrix}, \label{eqma}
\end{equation}
where $B_{i,i}$ are full-rank $s_i\times s_i$ matrices. BLTA equals the complete automorphisms of decreasing polar codes that can be formulated as affine transformations \cite{b8}.

\subsection{Successive cancellation decoding} \label{sec2C}
Let $L_{i,t}$ be the log likelihood ratio (LLR) of the $i$-th node at stage $t$ and $L_{i,m}$ be the received LLRs from channels. $L_{i,t}$ are propagated from stage $t+1$ according to
\[L_{i,t} = f(L_{i,t+1},L_{i+2^t,t+1}) = \log\left(\frac{e^{L_{i,t+1}+L_{i+2^t,t+1}}+1}{e^{L_{i,t+1}}+e^{L_{i+2^t,t+1}}}\right); \]
\begin{align*}
L_{i+2^t,t} &= g(u_{i,t},L_{i,t+1},L_{i+2^t,t+1})\\
&= (-1)^{u_{i,t}}L_{i,t+1}+L_{i+2^t,t+1}. 
\end{align*}

At stage $0$, we have
\begin{gather*}
u_{i,0} = \begin{cases}
    0 ,& \text{ if } i\in \MF; \\
    0 ,& \text{ if } i\in \MI, L_{i,0}\geq 0; \\
    1 ,& \text{ if } i\in \MI, L_{i,0}< 0. \\
\end{cases}
\end{gather*}

Then hard decisions $u_{i,t}$ are propagated from stage $t-1$ according to
\[u_{i,t} = u_{i,t-1}\oplus u_{i+2^t,t-1};\]
\[u_{i+2^t,t} = u_{i+2^t,t-1}.\]
where $\oplus$ means the addition modulo $2$.

Let $\text{SC}_{\MI}:\RR^n\to\FF_2^n$ map the received LLR vector $y = (L_{i,m})_{i\in\{0,...,n-1\}}\in\RR^n$ to the SC decoding result $ (u_{i,m})_{i\in\{0,...,n-1\}}=\text{SC}_{\MI}(y)\in C(\MI)$.

\subsection{Automorphism ensemble decoding}

Let $\pi_1,...,\pi_t$ be $t$ different automorphisms of the code $C$ and $y\in\RR^n$ be the received LLR vector. A list of decoders can independently decode each permuted LLR $\pi_j(y)$. The decoded candidate codeword of $y$ using $\pi_j$ is
\[\hat{x}_j = \pi_j^{-1}(\text{SC}_{\MI}(\pi_j(y)).\]
A final decoding result is selected according to the minimum Euclidean distance rule:
\[x = \arg\min_{\hat{x}_j,j=1,...,t}||\hat{x}_j-y||.\]
For an automorphism $\pi$ of $C(\MI)$,  we say $\pi$ commutes with $\text{SC}_{\MI}$ if for all $y\in\RR^n$, $\text{SC}_{\MI}(\pi(y))= \pi(\text{SC}_{\MI}(y))$. If $\pi$ commutes with $\text{SC}_{\MI}$, the corresponding permuted SC decoder always outputs the same decoding result as the non-permuted SC decoder.

The automorphism $\pi$ in LTA group is SC-invariant for $C(\MI)$, which means it commutes with $\text{SC}_{\MI}$ \cite{b5}. Moreover, in \cite{b10}, two affine automorphisms $\pi,\pi'$ are called equivalent for $C(\MI)$, denoted by $\pi\sim_{\MI}\pi'$, if for all $y\in\RR^n$
\[\pi^{-1}(\text{SC}_{\MI}(\pi(y))) = \pi'^{-1}\text{SC}_{\MI}(\pi'(y)).\]
The equivalence classes are defined as
\[[\pi]_{\MI} =\{\pi' : \pi'\sim_{\MI}\pi\}.\] 
Let $\Ba$ be identity permutation, the equivalence class $[\Ba]_{\MI}$ consists of the complete affine automorphisms commuting with $\text{SC}_{\MI}$, which is an automorphism subgroup. The authors of \cite{b10} proved that BLTA$([2,1...,1])\subseteq [\Ba]_{\MI}$, that is, automorphisms in BLTA$([2,1...,1])$ commute with $\text{SC}_{\MI}$ of any decreasing monomial code $C(\MI)$ whose automorphism group includes BLTA$([2,1...,1])$. A natural question arises: are these the complete SC-invariant affine automorphisms?

\section{Analysis on SC-invariant automorphisms}

In this section, we give a necessary and sufficient condition to identify SC-invariant affine automorphisms for any specific code $C(\MI)$. The key technique in proofs is that the block lower-triangular structure of transformation matrix can be used to decompose the corresponding automorphism into shorter ones (detailed description is in Definition \ref{def1}). Therefore, $C(\MI)$ can be decomposed to shorter subcodes, and we can identify SC-invariant automorphisms inductively.

\subsection{Notations and definitions}

Let $M$ be a full-rank matrix. We say $M$ has the block lower-triangular structure $s(M)=\langle s_1,...,s_l \rangle$ if $M$ can be written as (\ref{eqma}) and none of $B_{i,i}$ can be written as a block lower-triangular matrix with more than one block. Define $S_t=\sum_{i=1}^{t-1} s_i$ for $2\leq t\leq l+1$ and $S_1=0$.

Define $[a,b]$ to be the integer set $\{i\in \NN|a\leq i\leq b\}$ for $a,b\in\NN$, and $M([a,b],[c,d])$ to be the corresponding submatrix of $M$.

The affine transformation $(M,0)\subseteq [\Ba]_{\MI}$ if and only if $(M,b)\subseteq [\Ba]_{\MI}$ for any $b\in\FF_2^m$. For convenience, define $\varphi(M)$ to be the permutation $(M,0)$ which permutes $a\in\FF_2^m$ to $Ma$.

Define $\text{Ind}_m  (a_{i_1}=c_{i_1},...,a_{i_t}=c_{i_t}) 
= \{(a_1,...,a_m)^T\in\FF_2^m|a_{i_j}=c_{i_j}, \forall j=1,...,t\}$
%\begin{align*}
%& Ind_m  (a_{i_1}=c_{i_1},...,a_{i_t}=c_{i_t}) \\
%= & \{(a_1,...,a_m)^T\in\FF_2^m|a_{i_j}=c_{i_j}, \forall j=1,...,t\}
%\end{align*}
to be a set of indices whose $i_1,...,i_t$-th bits are fixed to be $a_{i_1},...,a_{i_t}$. Define
\begin{align*}
&\MI(\text{Ind}_m(a_{i_1}=c_{i_1},...,a_{i_t}=c_{i_t})) \\ 
=&\{(a_1,...a_{i_1-1},a_{i_1+1},...,a_{i_t-1},a_{i_t+1},...,a_m)^T\in\FF_2^{m-t}|\\
&(a_1,...,a_m)^T\in \MI, a_{i_j}=c_{i_j}, \forall j=1,...,t\}
\end{align*}
to be the information set of length-$2^{m-t}$ subcode consisting of indices in  $\text{Ind}_m(a_{i_1}=c_{i_1},...,a_{i_t}=c_{i_t})$. $\MF(\text{Ind}_m(a_{i_1}=c_{i_1},...,a_{i_t}=c_{i_t}))$ is defined in the same way.

\begin{figure}[!t]
\centering
\includegraphics[width=3in]{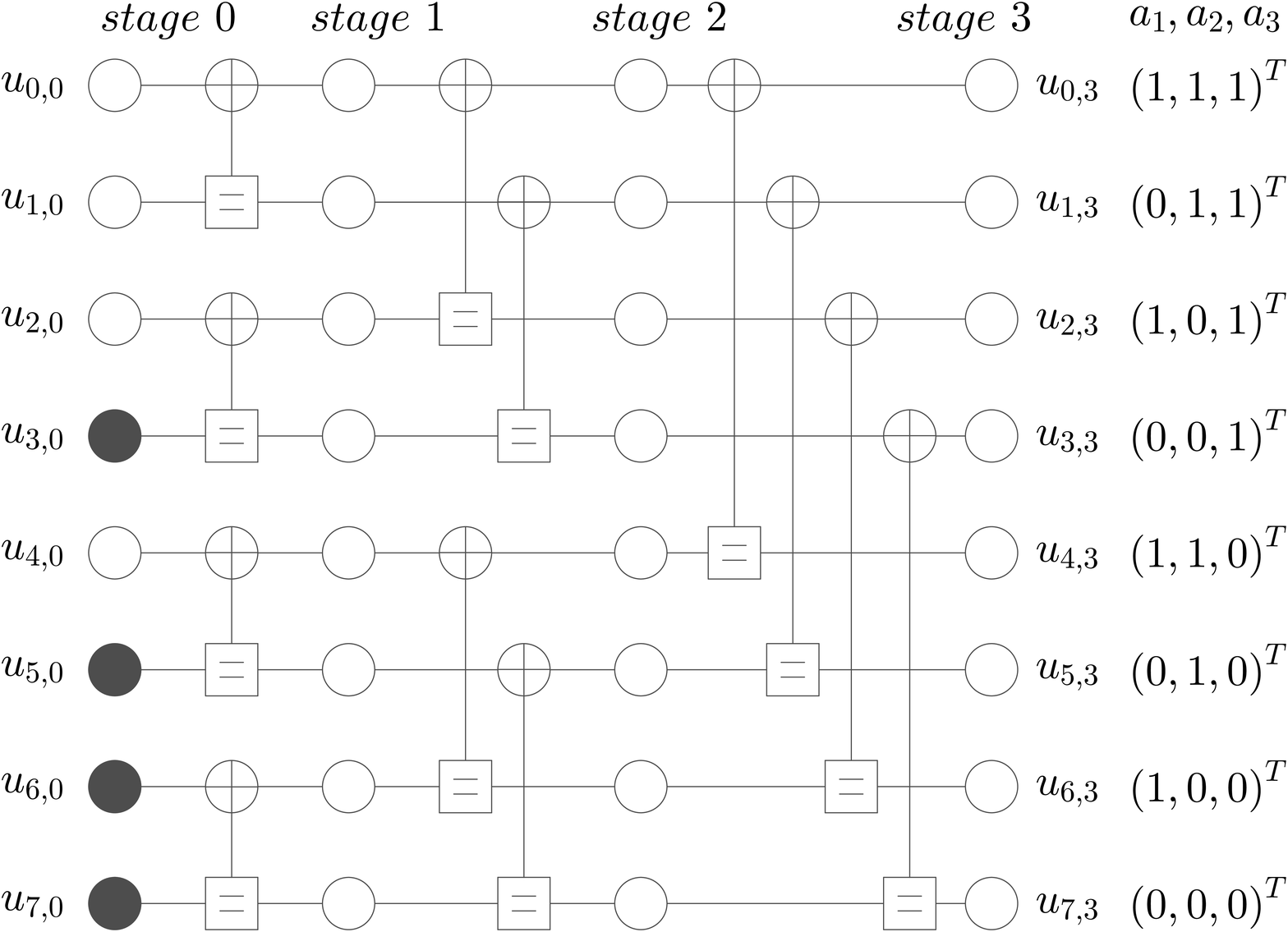}
\caption{The factor graph of an (8,4) polar code}
\label{fig2}
\end{figure}

For example, the factor graph of (8,4) polar code is shown in Fig. \ref{fig2}, where $\MI = \{3, 5, 6, 7\}$. In Fig. \ref{fig2}, $\MI(\text{Ind}_3(a_1=1)) = \{3\}$ (resp. $\MI(\text{Ind}_3(a_1=0))= \{1,2,3\}$) is the information set of length-4 subcode consisting of indices that the least significant bit $a_1$ is equal to $1$ (resp. $0$), i.e. the even (resp. odd) bits. Similarly, $\MI(\text{Ind}_3(a_3=1)) = \{3\}$ (resp. $\MI(\text{Ind}_3(a_3=0))= \{1,2,3\}$) is the information set of length-4 subcode consisting of indices that the most significant bit $a_3$ is equal to $1$ (resp. $0$), i.e. the first (resp. last) four bits. This definition simplifies the description of decomposed shorter subcodes in our proofs.

\subsection{Transformations of matrices}

In this subsection, we provide two lemmas on matrix transformations. 
 
\begin{lemma}[\textbf{lower-triangular transformation}]\label{lemma1}
Let $M$ be a full-rank matrix and $M_1, M_2$ be two lower-triangular matrices. $\pi = \varphi(M)$ and $\pi' = \varphi(M_1MM_2)$ are two automorphisms of $C(\MI)$. Then $\pi'$ commutes with $\text{SC}_{\MI}$ if and only if $\pi$ commutes with $\text{SC}_{\MI}$. Moreover, $s(M)= s(M_1MM_2)$.
\end{lemma}

\begin{proof}
It is from the fact that $\varphi(M_1MM_2) = \varphi(M_1)\varphi(M)\varphi(M_2)$ and $\varphi(M_1),\varphi(M_2)$ and $\varphi(M)$ all commute with $\text{SC}_{\MI}$.

Let $s(M) = \langle s_1,...,s_l\rangle$ and $s(M_1MM_2) = \langle s'_1,...,s'_k\rangle$. Notice that $L_1M$ means adding upper row of $M$ to lower, while $ML_2$ means adding right column of $M$ to left. Therefore, $M([1,s_1],[s_1+1,m]) = 0$ implies $M_1MM_2([1,s_1],[s_1+1,m]) = 0$, so $s'_1\leq s_1$. 

Since $M$ = $M_1^{-1}(M_1MM_2)M_2^{-1}$, similarly, we have $s_1\leq s'_1$. Therefore, $s_1=s'_1$. And so on, $s(M)= s(M_1MM_2)$.
\end{proof}

\begin{remark}\label{re1}
Thanks to Lemma \ref{lemma1}, we only need to investigate upper-triangular matrices since every matrix $M$ can be transformed to an upper-triangular matrix by lower-triangular transformation while maintaining the block lower-triangular structure. 
\end{remark}

Next, we show how to decompose upper-triangular transformation by exploiting its block lower-triangular structure.

\begin{definition}\label{def1}
Let $M$ be an upper-triangular matrix with $s(M)=\langle s_1,...,s_l\rangle$ and $\pi= \varphi(M)$. We have $M=M_1M_2$ where 
\begin{gather*}
M_1(i,j) = \begin{cases}
    M(i,j) ,& \text{ if } 1\leq i,j\leq S_l; \\
    1 ,& \text{ if } S_l+1\leq i =j\leq m; \\
    0 ,& \text{ otherwise }. \\
\end{cases}
\end{gather*}
And
\begin{gather*}
M_2(i,j) = \begin{cases}
    M(i,j) ,& \text{ if }  S_l+1\leq i,j\leq m; \\
    1 ,& \text{ if } 1\leq i=j\leq S_l; \\
    0 ,& \text{ otherwise }. \\
\end{cases}
\end{gather*}
We define four permutations related to $\pi$: $\pi_1 = \varphi(M_1), \pi_2 = \varphi(M_2), \tilde{\pi}_1 = \varphi(M([1,S_l],[1,S_l])),\tilde{\pi}_2 = \varphi(M([S_l+1,m],[S_l+1,m]))$.
\end{definition}

\begin{lemma}[\textbf{Permutation Decomposition}]\label{lemma2}
Let $M$ be an upper-triangular matrix with $s(M)=\langle s_1,...,s_l\rangle$ and $\pi= \varphi(M)$. Let $\pi_1,\pi_2,\tilde{\pi}_1,\tilde{\pi}_2$ be the permutations defined in Definition \ref{def1}. For $0\leq z_1\leq 2^{S_l}-1$ and $0\leq z_2\leq 2^{s_l}-1$,
\begin{equation}
\pi(z_1+2^{S_l}z_2) = \pi_1(z_1)+\pi_2 (2^{S_l}z_2), \label{eq21}
\end{equation}
\begin{equation}
\pi(z_1+2^{S_l}z_2) = \tilde{\pi}_1(z_1)+2^{S_l}\tilde{\pi}_2(z_2). \label{eq22}
\end{equation}

Moreover, if $s_l=1$, then $\pi=\pi_1$ and $\pi_2$ is the identical permutation, which means
\begin{equation}
\pi(z_1+2^{m-1}z_2) = \pi(z_1)+2^{m-1}z_2 \label{eq1}
\end{equation}
for $0\leq z_1\leq 2^{m-1}-1$ and $z_2=0,1$. And (\ref{eq1}) means $\pi([0,n/2-1])=[0,n/2-1]$ and $\pi([n/2,n-1])=[n/2,n-1]$, that is, bits in the upper (lower) half branch remain in the upper (lower) half branch after permutation.
\end{lemma}

\begin{proof}
Since $\pi=\pi_2\circ\pi_1$,  
\begin{align*}
& \pi(z_1+2^{S_l}z_2) \\
=& \pi_2\circ\pi_1(z_1+2^{S_l}z_2) \\
=& \pi_2(\pi_1(z_1)+2^{S_l}z_2)  \\
=& \pi_1(z_1)+\pi_2 (2^{S_l}z_2).
\end{align*} 
So (\ref{eq21}) is proved.  (\ref{eq22}) is from $\pi_1(z_1)=\tilde{\pi}_1(z_1)$ and $\pi_2(2^{S_l}z_2) = 2^{S_l}\tilde{\pi}_2(z_2)$. 
\end{proof}

\begin{remark}\label{re2}
Due to the block lower-triangular structure of $M$, $\pi=\pi_2\circ\pi_1= \pi_1\circ\pi_2$, where $\pi_1$ only affects the first $S_l$ bits and $\pi_2$ only affects the last $s_l$ bits. To be specific, permutation $\pi$ can be decomposed into two steps.

Step 1 (the effect of $\pi_1$): Divide $[0,2^m-1]$ into $2^{s_l}$ blocks $[i2^{S_l}, (i+1)2^{S_l}-1]$, $0 \leq i \leq 2^{s_l}-1$, then apply the same permutation $\tilde{\pi}_1$ to each block. 

Step 2 (the effect of $\pi_2$): Treat each block as a whole, apply $\tilde{\pi}_2$  to $2^{s_l}$ blocks. 

This technique will help us decompose the polar code into shorter subcodes in Algorithm \ref{alg:1}.
\end{remark}

%\begin{remark}\label{re3}
%Let $M$ be an upper triangle matrix with $s(M)=\langle m\rangle$, we can find two lower triangle matrices $L_1$ and $L_2$ such that $M'=L_1ML_2$ and $s(M'([1,m-1],[1,m-1]))=\langle m-1\rangle$ and $M'(m,[1,m])=[0,...,0,1]$.

%Define $k=\min\{i|M(i,m)=1\}$. If $M(k,m-1)=1$, $M'=M$. We note $s(M'([1,m-1],[1,m-1]))=\langle m-1\rangle$, otherwise if $s(M'([1,m-1],[1,m-1]))=[s_1,...,s_k]$, then $M'([1,s_1],[s_1,m-1])=0$. From construction of $M'$ we have $U([1,s_1],[s_1,m])=0$, which is a contradiction against $s(U)=\langle m\rangle$. 

%Now assume $M(k,m-1)=0$. If $M(m-1,m)=0$, first add the $m$-th column of $M$ to $m-1$-th column, then add the $m-1$-th row to $m$-th row and get a new upper triangle matrix $M'$ satisfies $M'(k,m-1)= 1$. 

%If $M(m-1,m)=1$, first add the $m$-th column of $M$ to $m-1$-th column, then add the $k$-th row of $M$ to $m-1$-th row, next add the $m-1$-th row of $M$ to $m$-th row, we get $M'$ satisfying $s(M'([1,m-1],[1,m-1]))=\langle m-1\rangle$. Finally add the $m$-th column of $M'$ to the first $m-2$ rows to satisfy $s(M'([1,m-1],[1,m-1]))=\langle m-1\rangle$ and $M'(m,[1,m])=[0,...,0,1]$. Notice that the final step does not influence $s(M'([1,m-1],[1,m-1]))$ since the value $\max\{i|M'(j,i)=1\}$ does not change for $i\leq k$ in the final step.
%\end{remark}

\subsection{Distinguishing SC-invariant automorphisms}

Algorithm \ref{alg:1} determines whether affine automorphisms with the block lower-triangular structure $\langle s_1,...,s_l\rangle$ commute with $\text{SC}_{\MI}$  iteratively. We claim that $\pi= \varphi(M)$ commutes with $\text{SC}_{\MI}$ if and only if DecAut$(s(M), \MI)$ outputs TRUE. Let $\pi_1,\pi_2,\tilde{\pi}_1,\tilde{\pi}_2$ be the permutations defined in Definition \ref{def1}. Note that $s(\tilde{\pi}_1) = \langle s_1,...,s_{l-1}\rangle$. We briefly describe procedures of Algorithm \ref{alg:1}. 

First, if $l=1$ (lines 5-7), $\pi$ is SC-invariant if and only if the code belongs to Rate-0, single parity check (SPC), repetition (Rep) or Rate-1 codes\cite{b14}. 

If $l\neq 1$, we recursively determine whether $\pi$ is SC-invariant. According to $s_l$, we consider two cases:

1) $s_l=1$ (lines 8-11), because $\pi_2$ is the identical permutation, $\pi$ is SC-invariant if and only if $\tilde{\pi}_1$ commutes with $C(\MI(A_1))$ and $C(\MI(A_2))$, i.e., the subcodes on upper half branch and lower half branch.

2) $s_l>1$ (lines 12-23), divide $[0,2^m-1]$ into $2^{s_l}$ blocks $[i2^{S_l}, (i+1)2^{S_l}-1]$, $0 \leq i \leq 2^{s_l}-1$. In this case, $\tilde{\pi}_2$ is not identical permutation, then $\pi$ is SC-invariant only if all frozen bits belong to the first block or all information bits belong to the last block. Furthermore, $\tilde{\pi}_1$ must commute with either the first subcode $C(\MI(A_1))$ (lines 13-14) or the last subcode $C(\MI(A_2))$ (lines 16-17) respectively.

\begin{example}\label{ex1}
Assume $C(\MI)$ is a polar code with length $n=16$ and information set $\MI = \{3,5,6,7,9,10,11,12,13,14,15\}$. It is clear that Aut$(C(\MI)) = \text{BLTA}([4])$. We can determine whether $\varphi(M)$ with $s(M)=\langle 3,1\rangle$ commutes with $SC_{\MI}$ by Algorithm \ref{alg:1}. Since $s_2=1$, from lines 8-11, $\text{DecAut}(\langle 3,1\rangle, \MI) = \text{DecAut}(\langle 3\rangle, \{3,5,6,7\})\wedge \text{DecAut}(\langle 3\rangle, \{1,2,3,4,5,6,7\})$. Next, $\text{DecAut}(\langle 3\rangle, \{3,5,6,7\})) = \text{FALSE}$ and $\text{DecAut}(\langle 3\rangle, \{1,2,3,4,5,6,7\})=\text{TRUE}$ from line 6. Therefore, we conclude that the algorithm will output FALSE so that $\varphi(M)$ with $s(M)=\langle 3,1\rangle$ does not commute with $SC_{\MI}$.
\end{example}

\begin{figure}[!t]
\begin{algorithm}[H]
\caption{DecAut$(\langle s_1,...,s_l\rangle, \MI)$}
\begin{algorithmic}[1]\label{alg:1}

\renewcommand{\algorithmicrequire}{\textbf{Input:}}
\renewcommand{\algorithmicensure}{\textbf{Output:}}
\REQUIRE block lower-triangular structure $\langle s_1,...,s_l\rangle$, information set $\MI$
\ENSURE  $a$ is a boolean value and $a$ is TRUE if and only if automorphisms with the block lower-triangular structure $\langle s_1,...,s_l\rangle$ commute with $\text{SC}_{\MI}$.
\STATE $m\gets \sum_{i=1}^l s_i$; $S_l\gets \sum_{i=1}^{l-1} s_i$;
\STATE $\MF\gets \{0,...,2^m-1\}/\MI$;
\STATE $A_1\gets \text{Ind}_m(a_{S_l+1}=1,...,a_m=1)$;
\STATE $A_2\gets \text{Ind}_m(a_{S_l+1}=0,...,a_m=0)$;
\IF{$l = 1$}
 \STATE $a \gets (\MF\subseteq A_1) \vee (\MI\subseteq A_2)$;
\ELSE
  \IF{$ s_l = 1$}
    \STATE $a_1 \gets \text{DecAut}(\langle s_1,...,s_{l-1}\rangle, \MI(A_1))$;
    \STATE $a_2 \gets \text{DecAut}(\langle s_1,...,s_{l-1}\rangle, \MI(A_2))$;
    \STATE $a \gets a_1\wedge a_2$;
  \ELSE
    \IF{$\MF\subseteq A_1$}
       \STATE $a \gets \text{DecAut}(\langle s_1,...,s_{l-1}\rangle, \MI(A_1))$;
    \ELSE
       \IF{$\MI\subseteq A_2$} 
           \STATE $a\gets \text{DecAut}(\langle s_1,...,s_{l-1}\rangle, \MI(A_2))$;
        \ELSE
            \STATE $a\gets$ FALSE;
        \ENDIF
     \ENDIF
   \ENDIF
\ENDIF
\end{algorithmic}
\end{algorithm}
\end{figure}

In Theorem \ref{thm1} and Theorem \ref{thm2}, we prove the sufficiency and necessity of Algorithm \ref{alg:1}, respectively.

For convenience, denote by $L_{i,t}$ and $u_{i,t}$ the LLRs and hard decisions of the $i$-th node at stage $t$ before permutation, and $L'_{i,t}$ and $u'_{i,t}$ the LLRs and hard decisions after permutation (see Section \ref{sec2C}).

\begin{theorem}[\textbf{Sufficiency}]\label{thm1}
Let $M$ be a block lower-triangular matrix with $s(M)=\langle s_1,...,s_l\rangle$ and $\pi= \varphi(M)$ is an automorphism of $C(\MI)$ with length $n=2^m$. $\pi$ commutes with $\text{SC}_{\MI}$ if DecAut$(s(M), \MI)$ outputs TRUE.
\end{theorem}

\begin{proof}
From Remark \ref{re1}, assume $M$ is an upper-triangular matrix. We prove the theorem by induction on $l$. If $l=1$, the theorem holds since the condition implies the code is one of Rate-0, SPC, Rep or Rate-1 code, where SC decoding is equivalent to ML decoding, and thus invariant under permutations, and any permutation produces the same decoding result.

Let $\pi_1,\pi_2,\tilde{\pi}_1,\tilde{\pi}_2$ be the permutations defined in Definition \ref{def1}. For the induction step $l-1\to l$, There are two cases we need to consider. The first is $s_l=1$. Divide $[0,2^{m-1}]$ into upper half branch $[0, 2^{m-1}-1]$ and lower half branch $[2^{m-1}, 2^m-1]$. As mentioned in Remark \ref{re2}, $\tilde{\pi}_2$ is the identical permutation so bits in the upper or lower half branch remain in the same branch after permutation. Thus, the LLRs at stage $m-1$ of the upper and lower half branches are permuted by $\tilde{\pi}_1$, respectively. Then by induction, $\tilde{\pi}_1$ is SC-invariant. It follows that $\pi$ is SC-invariant.

Now let us discuss the proof in detail. First consider the upper half branch. Note that $\pi(i)=\tilde{\pi}_1(i)$ for $0\leq i\leq 2^{m-1}-1$. (\ref{eq1}) implies $L'_{i+2^{m-1},m} =  L_{\pi(i+2^{m-1}),m} = L_{\tilde{\pi}_1(i)+2^{m-1},m}$ and then 
\begin{align*}
L_{\tilde{\pi}_1(i),m-1} &= f(L_{\tilde{\pi}_1(i),m},L_{\tilde{\pi}_1(i)+2^{m-1},m}) \\
& =f(L_{\pi(i),m},L_{\pi(i+2^{m-1}),m}) \\
& =f(L'_{i,m} + L'_{i+2^{m-1},m}) =L'_{i,m-1}.
\end{align*}
Because that DecAut$(\langle s_1,...,s_{l-1},1\rangle, \MI)$ outputs TRUE implies that DecAut$(\langle s_1,s_2,...,s_{l-1}\rangle, \MI(\text{Ind}_m(a_m=1)))$ outputs TRUE, by inductive hypothesis, 
\begin{equation}
u'_{i,m-1} = u_{\tilde{\pi}_1(i),m-1} = u_{\pi_1(i),m-1}. \label{eq31}
\end{equation}

Next, we consider the lower half branch. For $0\leq i\leq 2^{m-1}-1$ 
\begin{align*}
L_{\tilde{\pi}_1(i)+2^{m-1},m-1} &= g(u_{\tilde{\pi}_1(i),m-1},L_{\tilde{\pi}_1(i),m},L_{\tilde{\pi}_1(i)+2^{m-1},m}) \\
& = g(u'_{i,m-1},L'_{i,m},L'_{i+2^{m-1},m}) \\
& = L'_{i+2^{m-1},m-1}.
\end{align*}
Because that DecAut$(\langle s_1,...,s_{l-1},1\rangle, \MI)$ outputs TRUE implies that DecAut$(\langle s_1,s_2,...,s_{l-1}\rangle, \MI(\text{Ind}_m(a_m=0)))$ outputs TRUE, by inductive hypothesis, 
\begin{equation}
u'_{i+2^{m-1},m-1} = u_{\tilde{\pi}_1(i)+2^{m-1},m-1} = u_{\pi(i)+2^{m-1},m-1}. \label{eq32}
\end{equation} 

It follows from (\ref{eq31}) and (\ref{eq32}) that $u_{\pi(i),m} = u'_{i,m}$. 

Now we turn to the case $s_l> 1$. In this case, $\tilde{\pi}_2$ is not identical permutation, additional conditions are required to ensure SC-invariance of $\pi$. We divide $[0,2^m-1]$ into $2^{s_l}$ blocks $\text{Ind}_m(a_{S_l+1}=c_{S_l+1},...,a_m=c_m)$. By lines 6-10 of Algorithm \ref{alg:1}, either all frozen bits belong to the first block $A_1= \text{Ind}_m(a_{S_l+1}=1,...,a_{m}=1)$ or all information bits belong to the last block $A_2= \text{Ind}_m(a_{S_l+1}=0,...,a_m=0)$. 

1) If $\MF\subseteq A_1$, from Lemma \ref{lemma2}, we have $L_{\pi_1(i),S_l} = L'_{i,S_l}$ for $i\in A_1$. Notice that DecAut$(\langle s_1,s_2,...,s_l\rangle, \MI)$ outputs TRUE and $\MF\subseteq A_1$ imply that DecAut$(\langle s_1,s_2,...,s_{l-1}\rangle, \MI(A_1))$ outputs TRUE. Then by inductive hypothesis, $u_{\pi_1(i),S_l} = u'_{i,S_l}$ for $i\in A_1$. Notice that $\text{Ind}_m(a_{S_l+1}=c_{S_l+1},...,a_m=c_m)\subseteq \MI$ for $a_{S_l+1},...,a_m$ are not all one, then $u_{i,S_l} = \text{sign}(L_{i,S_l})$ for $i\notin A_1$.

Then $C(\MI)$ can be viewed as $2^{S_l}$ independent length-$2^{s_l}$ SPC codes. Define $A' = \text{Ind}_m(a_1=c_1,...,a_{S_l}=c_{S_l})$ and $\tilde{y}= (L'_{i,m})_{i\in A'}=(L_{\pi(i),m})_{i\in A'}=\tilde{\pi}_2(L_{\pi_1(i),m})_{i\in A'}$, then $(u'_{i,m})_{i\in A'} = \text{SC}_{\MI'}(\tilde{y})$ where $\MI'=\{1,...,2^{s_l}-1\}$ and the first bit is frozen to $u'_{z_1,S_l}$ with $z_1=(a_1,...,a_{S_l},1,...,1)^T$. That is, $(u'_{i,m})_{i \in A'}$ can be decoded as a length-$2^{s_l}$ SPC code with LLR vector $\tilde{y}$.

Since $\tilde{\pi}_2$ commutes with $\MI'$, we have
\begin{align*}
(u'_{i,m})_{i\in A'} & = \text{SC}_{\MI'}(\tilde{y}) = \text{SC}_{\MI'}(\tilde{\pi}_2(L_{\pi_1(i),m})_{i\in A'})\\
&=\tilde{\pi}_2(\text{SC}_{\MI'}(L_{\pi_1(i),m})_{i\in A'})= \tilde{\pi}_2(u_{\pi_1(i),m})_{i\in A'} \\
&= (u_{\pi(i),m})_{i\in A'},
\end{align*}
thus $u'_{i,m} = u_{\pi(i),m}$.

2) If $\MI\subseteq A_2$ we have $u_{i,S_l} = u'_{i,S_l} = 0$ for $i\notin A_2$. Thus,
\begin{equation}
u_{i,m} = u_{j,S_l}; u'_{i,m} = u'_{j,S_l}. \label{eq4}
\end{equation}
for $j\in A_2$ and $j\equiv i \mod 2^{S_l}$. From Lemma \ref{lemma2}, $L_{\pi_1(j),S_l} = L'_{j,S_l}$ for $j\in A_2$. Notice that DecAut$(\langle s_1,s_2,...,s_l\rangle, \MI)$ outputs TRUE and $\MI\subseteq A_2$ imply that DecAut$(\langle s_1,s_2,...,s_{l-1}\rangle, \MI(A_2))$ outputs TRUE. By inductive hypothesis, 
\begin{equation}
u_{\pi_1(j),S_l} = u'_{j,S_l}. \label{eq5}
\end{equation}  
Then
\[ u'_{i,m} = u'_{j,S_l} = u_{\pi_1(j),S_l}=u_{\pi(i),m},\]
where $j\equiv i \mod 2^{S_l}$ and $j\in A_2$. Here the first equation is from (\ref{eq4}), the second equation is from (\ref{eq5}), and the last is because of (\ref{eq4}) and $\pi_1(j)\equiv \pi(j)\equiv \pi(i) \mod 2^{S_l}$ from Lemma \ref{lemma2}.
\end{proof}

The next lemma allows us to claim an automorphism is not SC-invariant by decomposing the automorphism on the upper and lower half branches even if $s_l\neq 1$. It will help us prove the necessity. 

\begin{lemma}\label{lemma3}
Let $C(\MI)$ be a decreasing monomial code with length $n=2^m$. $\pi=\varphi(M)$ is an automorphism of $C(\MI)$, where $M$ is an upper-triangular matrix. Let $A_i = \text{Ind}_m(a_m=i)$ and $\MI_i =\MI(A_i)$, $i = 0,1$, denote $\pi' = \varphi(M([1,m-1],[1,m-1]))$, then $\pi$ commutes with $\text{SC}_{\MI}$ implies $\pi'$ commutes with $\text{SC}_{\MI_1}$ and $\text{SC}_{\MI_0}$, i.e., $\pi'$ commutes with the subcodes on upper and lower half branches.
\end{lemma}

\begin{proof}
If $\pi'$ does not commute with $\text{SC}_{\MI_1}$, because $\pi'' \triangleq \pi|_{A_1} = (M([1,m-1],[1,m-1]),M([1,m-1],m))$, $\pi''$ does not commute with $\text{SC}_{\MI_1}$ as well. So there exists $y \in \RR^{2^{m-1}}$ such that $\pi''(\text{SC}_{\MI_1}(y))\neq \text{SC}_{\MI_1}(\pi''(y))$. 

Now we can construct an example from $y$ to show $\pi$ does not commute with $\text{SC}_{\MI}$. To be specific, let $(L_{i,m})_{i\in A_1}=y$ and $(L_{i,m})_{i \in A_0}= + \infty$, then $L_{i,m-1} = f(L_{i,m},+\infty) = L_{i,m}$ for $i \in A_1$. Therefore, $(L_{i,m-1})_{i\in A_1} = y$ and $(L'_{i,m-1})_{i\in A_1} = \pi''(y)$. Since $\pi''(\text{SC}_{\MI'}(y))\neq \text{SC}_{\MI'}(\pi''(y))$, we have for some $j$
\begin{equation}
u_{\pi''(j),m-1}\neq u'_{j,m-1}. \label{eq2}
\end{equation}
For $i \in A_1$, $L_{i+2^{m-1},m-1} = g(u_{i,m-1}, L_{i,m},+ \infty) = + \infty$. Similarly, $L'_{i+2^{m-1},m-1} = + \infty$. Thus
\begin{equation}
u_{i+2^{m-1},m-1} = u'_{i+2^{m-1},m-1} = 0. \label{eq3}
\end{equation}
Together with (\ref{eq2}) and (\ref{eq3}), $u_{\pi(j),m}\neq u'_{j,m}$ for some $j$. 

$\pi'$ commutes with $\text{SC}_{\MI_0}$ can be proved similarly when $(L_{i,m})_{i \in A_1} = \varepsilon$ and $(L_{i,m})_{i \in A_0}=y$, where $\varepsilon$ is positive and small enough.
\end{proof}

The next lemma proves two special cases of the necessity by decomposing the permutation on the subcodes consisting of odd and even indices. 

\begin{lemma}\label{lemma4}
Let $M$ be a block lower-triangular matrix with $s(M)=\langle 1,1,...,1,s_l\rangle$ where $s_l=2$ or $3$. $\pi= \varphi(M)$ is an automorphism of $C(\MI)$ with length $n=2^m$. Then $\pi$ commutes with $C(\MI)$ only if $\MF\subseteq \text{Ind}_m(a_{S_l+1}=1,...,a_m=1)$ or $\MI\subseteq \text{Ind}_m(a_{S_l+1}=0,..,a_m=0)$.
\end{lemma}

\begin{proof}
We inducted on $m$, if $m = 2,3$, the lemma can be proved by exhaustive search. For the induction step $m-1\to m$, assume $s(M)=\langle 1,1,...,1,2\rangle$. Let $\MI$ be an information set such that $\MF\not\subseteq A_1 = \text{Ind}_m(a_{m-1}=1,a_m=1)$ and $\MI\not\subseteq A_2 = \text{Ind}_m(a_{m-1}=0,a_m=0)$. Divide $\MI$ into two information sets $\MI_1=\MI(\text{Ind}_m(a_1=1))$ on the even bits and $\MI_2=\MI(\text{Ind}_m(a_1=0))$ on the odd bits, then $\MF_1=\text{Ind}_m(a_1=1)-\MI_1$ and $\MF_2=\text{Ind}_m(a_1=0)-\MI_2$. 

We are going to show that at least one of $\MI_1$ and $\MI_2$ does not satisfy the condition. First, $\MF_1\not\subseteq A_1$, since $\MF_1\subseteq A_1$ implies $\MF_2\subseteq A_1$ by decreasing property, which is contradictory against $\MF\not\subseteq A_1$. Similarly, $\MI_2\not\subseteq A_2$.   

We claim that one of $\MF_2\not\subseteq A_1$ and $\MI_1\not\subseteq A_2$ must hold, since otherwise $\text{Ind}_m(a_1=0,a_{m-1}=1,a_m=0)\subseteq \MI$ and $ \text{Ind}_m(a_1=1,a_{m-1}=1,a_m=0)\subseteq\MF$. Since $m>4$,  $\text{Ind}_m(a_1=0,a_2=1,a_{m-1}=1,a_m=0)\subseteq \MI$ and $\text{Ind}_m(a_1=1,a_2=0,a_{m-1}=1,a_m=0)\subseteq\MF$, which are contradictory against $\MI$ is a decreasing set when $m\geq 4$. 

Now we are going to construct a counter-example by induction. Assume $\MI_1\not\subseteq A_2$, denote $\pi' = \varphi(M([2,m],[2,m]))$. From inductive hypothesis, there exists some $\tilde{y}\in\RR^{2^{m-1}}$ such that $\pi'(\text{SC}_{\MI_1}(\tilde{y}))\neq \text{SC}_{\MI_1}(\pi'(\tilde{y}))$, which implies $\varphi(M) $ does not commute with $\text{SC}_{\MI}$ by setting $(L_{i,m})_{i\in \text{Ind}_m(a_1=1)}=\tilde{y}$ and $L_{i,m} = +\infty$ otherwise. If $\MF_2\not\subseteq A_1$, denote $(L_{i,m})_{i\in \text{Ind}_m(a_1=0)}=\tilde{y}$ and $L_{i,m} =  \varepsilon$ where $\varepsilon$ is positive and small enough otherwise.

If $s(M)=\langle 1,1,...,1,3\rangle$, the proof is similar if we take $A_1 = \text{Ind}_m(a_{m-2}=1,a_{m-1}=1,a_m=1)$ and $ A_2 = \text{Ind}_m(a_{m-2}=0,a_{m-1}=0,a_m=0)$.
\end{proof}

Now we are ready to prove Theorem \ref{thm2}.

\begin{theorem}[\textbf{Necessity}]\label{thm2}
Let $M$ be a block lower-triangular matrix with $s(M)=\langle s_1,...,s_l\rangle$ and $\pi= \varphi(M)$ is an automorphism of $C(\MI)$ with length $n=2^m$. $\pi$ commutes with $\text{SC}_{\MI}$ only if DecAut$(s(M), \MI)$ outputs TRUE.
\end{theorem}

\begin{table*}[htbp] %%d
\begin{center}
\setlength{\tabcolsep}{1.5mm}{
\begin{tabular}{c|c|c|c|c|c}
\hline
$(n,K)$	& $\MI_{\text{min}}$& affine automorphism group	& $[\Ba]_{\MI}$	& SC-invariant permutations in \cite{b10} & SC-invariant permutations in this paper
 \\
\hline
$(256,128)$	& $\{31,57\}$ & BLTA$([3,5])$ &	BLTA$([3,1,1,1,1,1])$ &	$3\times 2^{28}$ & $21\times 2^{28}$\\

$(128,85)$ & $\{23,25\}$ & BLTA$([3,1,3])$ & BLTA$([3,1,1,1,1])$ & $3\times 2^{21}$ &	$21\times 2^{21}$\\

$(64,32)$ &	$\{24\}$ & BLTA$([3,3])$ & BLTA$([3,2,1])$ & $3\times 2^{15}$ &$63\times 2^{15}$\\
\hline
\end{tabular}}
\\

\caption{The number of SC-invariant permutations for certain codes}
\label{tab1}
\end{center}
\end{table*}

\begin{proof}
From Remark \ref{re1}, assume $M$ is an upper-triangular matrix. We prove the theorem by induction on $m$. if $m\leq 3$, it can be proved by computer search. Assume the theorem holds for $m'\leq m-1$. Define $A_1 = \text{Ind}_m(a_m=1)$ and $A_0 = \text{Ind}_m(a_m=0)$.  Cases are classified according to $s_l$.

If $s_l=1$, the theorem can be proved by Lemma \ref{lemma3}.

If $s_l=2$ or $3$, Let $\pi_1,\pi_2,\tilde{\pi}_1,\tilde{\pi}_2$ be the permutations defined in Definition \ref{def1}. Divide $[0,2^m-1]$ into $2^{s_l}$ blocks $\text{Ind}_m(a_{S_l+1}=c_{S_l+1},...,a_m=c_m)$. Denote $\MI'=\MI(\text{Ind}_m(a_{S_l+1}=c_{S_l+1},...,a_m=c_m))$. Since $\pi$ is SC-invariant, repeatedly applying Lemma \ref{lemma3} reveals that$\tilde{\pi}_1$ commutes with $\text{SC}_{\MI'}$. By Theorem \ref{thm1}, $\pi_1$ commutes with $\text{SC}_{\MI}$. Therefore, $\pi_2=\pi_1^{-1}\circ\pi$ commutes with $\text{SC}_{\MI}$. Then the theorem can be proved by Lemma \ref{lemma4}.

If $s_l\geq 4$, without loss of generality, assume $s(M([1,m-1],[1,m-1]))=\langle s_1,...,s_{l-1},s_l-1\rangle$ and $M(m,[1,m-1])=0$. Define $\pi' = \varphi(M([1,m-1],[1,m-1]))$. (This can be achieved by transformations in Lemma \ref{lemma1}.) From Lemma \ref{lemma3}, $\pi$ commutes with $\text{SC}_{\MI}$ only if $\pi'$ commutes with $\text{SC}_{\MI(A_i)}$ for $i=0,1$.

From inductive hypothesis, for all $i = 0,1$, one of $\MF(A_i)\subseteq  \text{Ind}_m(a_{S_l+1}=1,...,a_{m-1}=1,a_m=i)$ and $\MI(A_i)\subseteq \text{Ind}_m(a_{S_l+1}=0,...,a_{m-1}=0,a_m=i)$ holds. Now we are going to prove one of $\MF\subseteq \text{Ind}_m(a_{S_l+1}=1,...,a_{m-1}=1,a_m=1)$ and $\MI\subseteq \text{Ind}_m(a_{S_l+1}=0,...,a_{m-1}=0,a_m=0)$ must hold. We consider the following three cases:

1) If $\MI(A_1)=\varnothing$, then $A_1\subseteq \MF$. Since $\varphi(M)$ with $s(M)=\langle s_1,...,s_l\rangle$ is an automorphism of $C(\MI)$, for any permutations $\rho\in sym(m)$ that permutes $[S_j+1,S_{j+1}]$ to $[S_j+1,S_{j+1}]$ for $1\leq j\leq l$, $(a_1,...,a_m)\in\MI$ is equal to $(a_{\rho(1)},...,a_{\rho(m)})^T\in\MI$. Therefore, $A_1\subseteq \MF$ implies $[n/2,n-2^{S_l}-1]\subseteq \MF$. Thus, $\MI\subseteq \text{Ind}_m(a_{S_l+1}=0,...,a_{m-1}=0,a_m=0)$ must hold. 

2) If $\MF(A_0) = \varnothing$, similarly, $\MF\subseteq \text{Ind}_m(a_{S_l+1}=1,...,a_{m-1}=1,a_m=1)$ must hold.

3) If $\MI(A_1)\neq\varnothing$ and $\MF(A_0) \neq\varnothing$. By properties of affine automorphism group, $\MI(A_1)\neq\varnothing$ implies $\MI(A_0)\not\subseteq \text{Ind}_m(a_{S_l+1}=0,...,a_{m-1}=0,a_m=0)$. Thus 
\[\MF(A_0)\subseteq \text{Ind}_m(a_{S_l+1}=1,...,a_{m-1}=1,a_m=0). \]
Similarly, 
\[\MI(A_1)\subseteq \text{Ind}_m(a_{S_l+1}=0,...,a_{m-1}=0,a_m=1).\] 
Then $\text{Ind}_m(a_{m-3}=0, a_{m-2}=1, a_{m-1}=1, a_m=0)\subseteq \MI$ and $\text{Ind}_m(a_{m-3}=1, a_{m-2}=0, a_{m-1}=0, a_m=1)\subseteq\MF$, which is contradictory against automorphism group. 

\end{proof}

From Algorithm \ref{alg:1}, SC-invariance of $\varphi(M)$ only depends on the block lower-triangular structure of $M$. Thus, we can prove the following theorem.

\begin{theorem}\label{thm3}
$[\Ba]_{\MI}$ is in the form of BLTA. 
\end{theorem}

\begin{proof}
Let $M$ be a block lower-triangular matrix with $s(M)=\langle s_1,...,s_l\rangle$ satisfying $\pi= \varphi(M)$ commutes with $\text{SC}_{\MI}$ and $l$ is as small as possible. Then BLTA$([s_1,...,s_l])\subseteq [\Ba]_{\MI}$. Now assume BLTA$([s'_1,...,s'_k])\subset [\Ba]_{\MI}$ but BLTA$([s'_1,...,s'_k])\not\subset$ BLTA$([s_1,...,s_l])$. Then there must exist some $i,j$ such that $S'_i<S_j<S'_{i+1}$. Let $M_1$ be a permutation matrix which permutes $S'_i$ and $S_j$ and keeps the other positions invariable, then $\varphi(M_1)\in$ BLTA$([s'_1,...,s'_k])$. However, $s(M_1M) = \langle s_1,...,s_{j-2},s_{j-1}+s_j,s_{j+1},...,s_l\rangle$ and $\varphi(M_1M)\in [\Ba]_{\MI}$, which is contradictory against that $l$ is as small as possible.
\end{proof}

In Algorithm \ref{alg:1}, we determine whether an affine automorphism commutes with $\text{SC}_{\MI}$. With Algorithm \ref{alg:2}, we further determine the complete SC-invariant affine automorphism group $[\Ba]_{\MI}=$ BLTA(DecGroup$(\MI,m))$.

\begin{figure}[!t]
\begin{algorithm}[H]
\caption{DecGroup$(\MI,m)$}
\begin{algorithmic}[1]\label{alg:2}

\renewcommand{\algorithmicrequire}{\textbf{Input:}}
\renewcommand{\algorithmicensure}{\textbf{Output:}}
\REQUIRE the information sets $\MI$, the code dimension $m$.
\ENSURE $s=[s_1,...,s_l]$; \# BLTA$([s_1,...,s_l])=[\Ba]_{\MI}$
\STATE  $\MF\gets \{0,...,2^m-1\}/\MI$;  
\IF{$m=0$}
  \STATE $s=[]$;
  \STATE return;
\ENDIF

\FOR{$t=m; t \geq 2; t--$}
  \STATE $A_1\gets \text{Ind}_m(a_{m-t+1}=1,...,a_m=1)$;
  \STATE $A_2\gets \text{Ind}_m(a_{m-t+1}=0,...,a_m=0)$;
  \IF{$\MF\subseteq A_1$}
    \STATE $s \gets [\text{DecGroup}(\MI(A_1),m-t),t]$;
    \STATE return;
  \ELSE
  \IF{$\MI\subseteq A_2$} 
      \STATE $s \gets [\text{DecGroup}(\MI(A_2),m-t),t]$;
      \STATE return;
    \ENDIF
  \ENDIF
\ENDFOR
\STATE $s'\gets[\text{DecGroup}(\MI(\text{Ind}_m(a_m=1),m-1),1]$;
\STATE $s''\gets[\text{DecGroup}(\MI(\text{Ind}_m(a_m=0),m-1),1]$;
\STATE $s \gets \text{Gro}(s',s'')$; \# BLTA$(s)$ =  BLTA$(s')\cap$ BLTA$(s'')$.
\end{algorithmic}
\end{algorithm}
\end{figure}

Without loss of generalization, assume $[\Ba]_{\MI} = \text{BLTA}([s_1,...,s_l])$. We first determine $s_l$, then $[s_1,...,s_{l-1}]$ can be obtained by calling the algorithm recursively. 

First, $s_l$ is determined by the loop in line 6. For $2 \leq t \leq m$, divide $[0,2^m-1]$ to $2^t$ blocks. We have $s_l = t$ if and only if $t$ is the largest integer such that all frozen bits belong to the first block $A_1=[0,2^{m-t}-1]$ (line 9) or all information bits belong to the last block $A_2=[2^m-2^{m-t},2^m-1]$ (line 13). If for all $2 \leq t \leq m$, the above conditions are not satisfied, we have $s_l=1$. 

If $s_l\geq 2$, $[s_1,...,s_{l-1}]$ can be recursively obtained by calling the algorithm with $\MI(A_1)$ (line 10) when $\MF\subseteq A_1$ or $\MI(A_2)$ (line 14) when $\MI\subseteq A_2$. If $s_l=1$, BLTA$[s_1,...,s_{l-1}]$ is the intersection of the SC-invariant affine automorphism groups of subcodes on the upper and lower half branches (lines 19-21). In line 21, Gro$(s',s'')$ output the array $s$ satisfying BLTA$(s)$ =  BLTA$(s')\cap$ BLTA$(s'')$. Such $s$ exists and can be found by the following lemma.

\begin{lemma}\label{lemma5}
The intersection of two BLTA groups is in the form of BLTA. 
\end{lemma}

\begin{proof}
We are going to find the BLTA group which is the intersection of BLTA$(s')$ and BLTA$(s'')$. Let $\{S_t\} = \{S_t'\}\cup \{S_t''\}$, and $s$ is induced by $\{S_t\}$, that is, $s_t =S_{t+1}-S_t$. Next we are going to prove BLTA$(s)$ =  BLTA$(s')\cap$ BLTA$(s'')$.

It is clear that BLTA$(s)\subseteq$ BLTA$(s')$ and BLTA$(s)\subseteq$ BLTA$(s'')$. Therefore, we only need to prove  BLTA$(s')\cap$ BLTA$(s'')\subseteq$ BLTA$(s)$. Assume $(M,b)\in\text{BLTA}(s')\cap\text{BLTA}(s'')$. Now we consider $M(j,k)$ for $S_i+1\leq j\leq S_{i+1}$ and $S_{i+1}+1\leq k\leq m$. Without loss of generality, assume $S_i=S_t'$, By the construction of $\{S_t\}$, we have $S_{i+1}\leq S_{t+1}'$. Then $M(j,k)=0$ since $(M,b)\in\text{BLTA}(s')$. Therefore, $(M,b)\in$ BLTA$(s)$.
\end{proof}

\begin{remark}
Algorithm \ref{alg:2} selects each $s_i$ as its largest possible value such that Algorithm \ref{alg:1} will not output FALSE, so it will output the complete SC-invariant affine automorphism group. Otherwise, if there exists another SC-invariant automorphism not in the output group, from Theorem \ref{thm3}, there will be a larger SC-invariant BLTA automorphism group with some larger $s_i$, which is a contradiction. Since the time complexity of one iteration is $O(m)$, the complexity of Algorithm \ref{alg:2} is $O(m2^m)=O(n\log n)$.
\end{remark}

\begin{example}
We now determine the complete SC-invariant affine automorphism group of $C(\MI)$ in Example \ref{ex1} by Algorithm \ref{alg:2}. For all $2\leq t\leq 4$, the conditions in line 9 and line 13 are not satisfied, so the last number of $s$ is 1. Then we call DecGroup$(\{3,5,6,7\},3)$ and DecGroup$(\{1,2,3,4,5,6,7\},3)$. DecGroup$ (\{3,5,6,7\},3)$ will output $[2,1]$ and DecGroup$(\{1,2,3,4,5,6,7\},3)$ will output $[3]$. Then $s=  \text{Gro}([2,1,1],[3,1]) = [2,1,1]$. Therefore, the complete SC-invariant affine automorphism group of $C(\MI)$ is BLTA$([2,1,1])$.
\end{example}

\section{Simulation}

\begin{figure}[!t]
\centering
\includegraphics[width=3.2in]{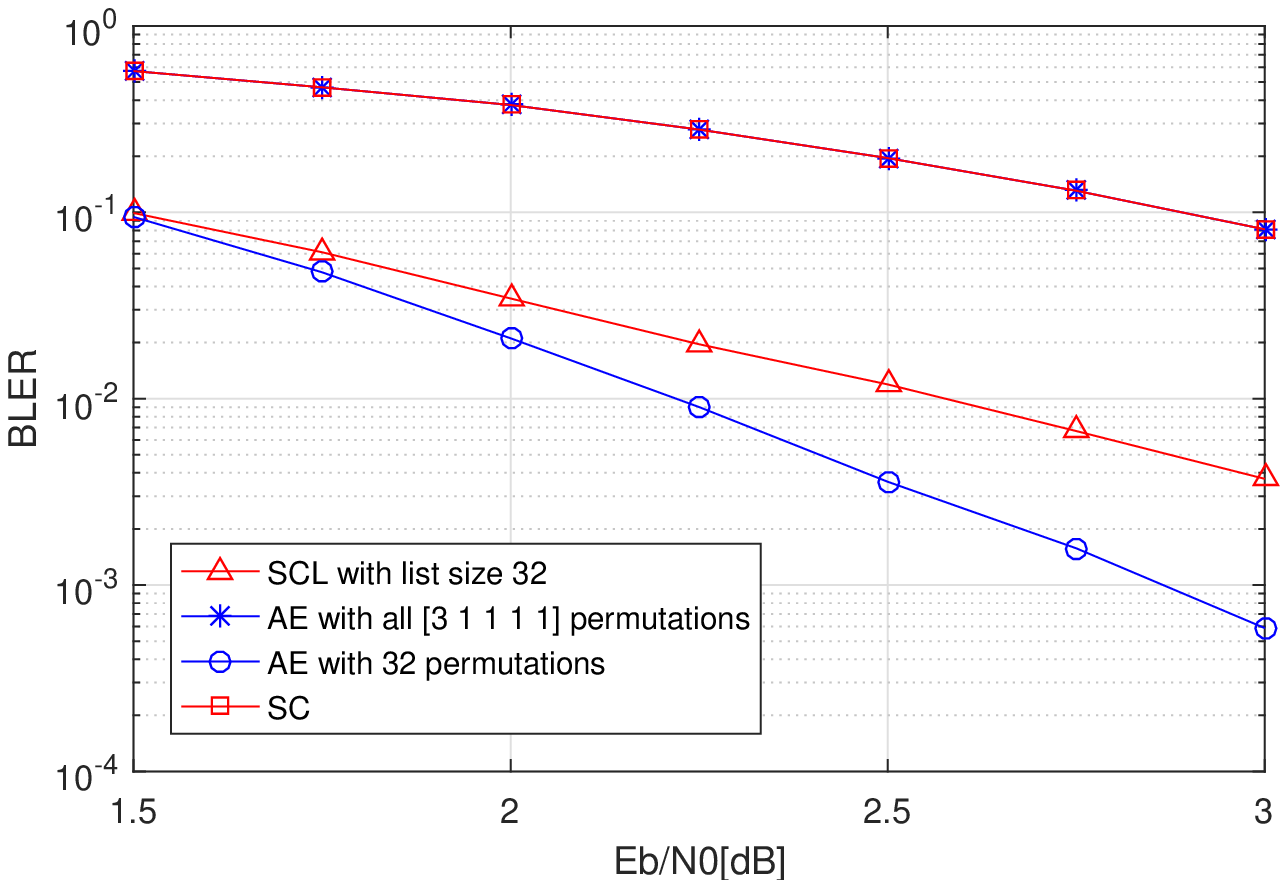}
\caption{Performance of a (256,128) polar code}
\label{fig1}
\end{figure}

Fig. \ref{fig1} shows the block error rate (BLER) performance of the (256,128) polar code studied in \cite{b7} and \cite{b10}. The code is generated by $\MI_{\text{min}} = \{31,57\}$ and has affine automorphism group \text{BLTA}$([3,5])$. In this case, $[\Ba]_{\MI} = \text{BLTA}([3,1,1,1,1,1])$, and it is shown that all the automorphisms in BLTA$([3,1,1,1,1,1])$ are futile in AE-SC decoding. Since the complete SC-invariant affine automorphisms are determined, the number of equivalent classes can be reduced from 68355 \cite{b10} to 9765. 

Table \ref{tab1} compares the number of SC-invariant affine automorphisms found in this paper with BLTA$([2,1...,1])$. Under several code constructions, the SC-invariant automorphism group can be larger than BLTA$([2,1...,1])$, which benefits applications requiring SC-invariant automorphisms. 

\section{Conclusion}
In this paper, we determine and prove the complete SC-invariant affine automorphisms for any specific decreasing polar code, which form a BLTA group. Compared to previous works, more SC-invariant affine automorphisms can be found according to our results. It helps us remove redundant permutations in AE-SC decoding and contributes to other applications requiring SC-invariant automorphisms.

\end{document}